\documentclass[aps,pra,reprint,floatfix,superscriptaddress]{revtex4-1}
\pdfoutput=1
\usepackage[ascii]{inputenc}
\usepackage[bookmarksnumbered,hypertexnames=false]{hyperref}
\usepackage{bbm,braket,microtype,mathrsfs,amsmath,amssymb,color,amsthm,graphicx}
\usepackage[T1]{fontenc}
\usepackage[USenglish]{babel}
\usepackage{times}
\usepackage{verbatim}
\usepackage[capitalize]{cleveref}

\newtheorem{thm}{Theorem}\crefname{thm}{Theorem}{Theorems}
\newtheorem{lem}[thm]{Lemma}\crefname{lem}{Lemma}{Lemmas}
\crefname{cor}{Corollary}{Corollaries}
\crefname{dfn}{Definition}{Definitions}
\DeclareMathOperator{\tr}{tr}

\DeclareMathOperator{\GHZ}{GHZ}
\DeclareMathOperator{\Cliff}{Cliff}

\DeclareMathOperator{\Stab}{Stab}
\DeclareMathOperator{\GL}{GL}
\DeclareMathOperator{\Sp}{Sp}

\newcommand{\CC}{\mathbb C}

\newcommand{\FF}{\mathbb F}

\newcommand{\ot}{\otimes}

\newcommand{\ssection}[1]{\smallskip\phantomsection\addcontentsline{toc}{section}{#1}\textit{#1.---}}

\allowdisplaybreaks[4]

\begin{document}

\title{Multipartite Entanglement in Stabilizer Tensor Networks}
\author{Sepehr Nezami}
\affiliation{Stanford Institute for Theoretical Physics, Stanford University, Stanford, USA}
\affiliation{Institute for Quantum Information and Matter and Walter Burke Institute for Theoretical Physics, California Institute of Technology, Pasadena CA 91125, USA}
\author{Michael Walter}
\affiliation{Stanford Institute for Theoretical Physics, Stanford University, Stanford, USA}
\affiliation{Korteweg-de Vries Institute for Mathematics, Institute for Theoretical Physics, Institute for Logic, Language, and Computation, and QuSoft, University of Amsterdam, The Netherlands}
\begin{abstract}
Despite the fundamental importance of quantum entanglement in many-body systems, our understanding is mostly limited to bipartite situations.
Indeed, even defining appropriate notions of multipartite entanglement is a significant challenge for general quantum systems.
In this work, we initiate the study of multipartite entanglement in a rich, yet tractable class of quantum states called \emph{stabilizer tensor networks}.
We demonstrate that, for generic stabilizer tensor networks, the \emph{geometry} of the tensor network informs the multipartite entanglement structure of the state.
In particular, we show that the average number of Greenberger-Horne-Zeilinger (GHZ) triples that can be extracted from a stabilizer tensor network is small, implying that tripartite entanglement is scarce.
This, in turn, restricts the higher-partite entanglement structure of the states.
Recent research in quantum gravity found that stabilizer tensor networks reproduce important structural features of the AdS/CFT correspondence, including the Ryu-Takayanagi formula for the entanglement entropy and certain quantum error correction properties.
Our results imply a new operational interpretation of the monogamy of the Ryu-Takayanagi mutual information and an entropic diagnostic for higher-partite entanglement.
Our technical contributions include a spin model for evaluating the average GHZ content of stabilizer tensor networks, as well as a novel formula for the third moment of random stabilizer states, which we expect to find further applications in quantum information.
\end{abstract}
\maketitle

Quantum entanglement is of fundamental relevance for the behavior of quantum mechanical systems in condensed matter and high energy physics.
From the perspective of quantum information processing, it is the resource that provides speedups in quantum computing, security in quantum cryptography, and improved performance in quantum sensing.
However, the structure of many-body or multipartite entanglement is only poorly understood~\cite{walter2016multipartite}.
In this work, we focus on analyzing multipartite entanglement in an important but tractable class of quantum states known as \emph{stabilizer tensor networks}, i.e., tensor networks that are obtained by contracting stabilizer states.
Stabilizer states are an important family of quantum states that can be highly entangled (even maximally so) but still have sufficient algebraic structure to admit an efficient classical description.
This makes them a versatile tool in quantum information theory, particularly in the theory of quantum error correction~\cite{gottesman96stabilizer}.
Of particular import in the present context is that the tripartite entanglement structure of stabilizer states can be precisely quantified -- any tripartite stabilizer state is locally equivalent to a collection of 
bipartite Bell pairs and tripartite GHZ states~\cite{bravyi2006ghz,looi2011tripartite} (cf.~\cite{hein2004multiparty,van2005invariants,van2005local,hein2006entanglement,smith2006typical,dahlsten2006exact,plato2008random}).

An important additional motivation to study stabilizer tensor networks comes from current research in quantum gravity.
In recent years, research in quantum gravity and quantum information theory has been inspired by a fruitful mutual exchange of ideas.
Tensor networks in particular provide a common framework, rooted in the similarity between the structure of the tensor network and the bulk geometry in holographic duality~\cite{swingle2012constructing,swingle2012entanglement,hartman2013time}.
A paradigmatic example is the Ryu-Takayanagui formula, $S(A) \simeq \lvert\gamma_A\rvert/4G_N$, which asserts that the entanglement entropy of a boundary region $A$ in a holographic state is in leading order proportional to the area of a corresponding minimal surface $\gamma_A$ in the bulk geometry~\cite{ryu2006holographic,lewkowycz2013generalized}.
Likewise, in any tensor network, the entanglement entropy of a boundary subsystem can be upper-bounded in terms of the size of a minimal cut through the network~\cite{orus2014practical} (\cref{fig:tns-and-rt}).
This bound can be saturated not only through the choice of suitable tensors~\cite{pastawski2015holographic,yang2015bidirectional} but is in fact a generic phenomenon in \emph{random tensor networks} with large bond dimension~\cite{hayden2016holographic,hastings2016asymptotics}, the mechanism of which can be understood in terms of multipartite entanglement distillation.
These tensor network models not only reproduce the Ryu-Takayanagi formula for the entanglement entropy, but they also implement several other significant features of holographic duality~\cite{pastawski2015holographic,yang2015bidirectional,hayden2016holographic}.
In many ways, these properties follow from the bipartite entanglement structure and can be therefore reduced to entropic considerations.

\begin{figure}
\includegraphics[height=3cm]{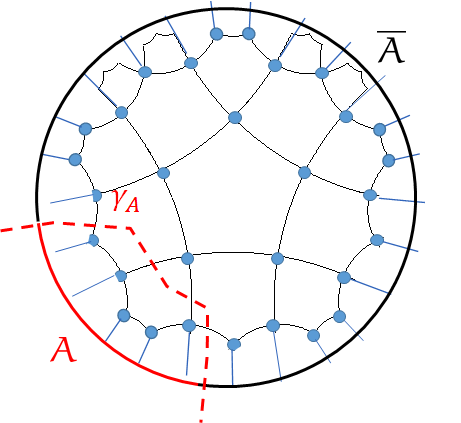}
\caption{\textbf{Stabilizer tensor networks.} A tensor network state is obtained by placing random stabilizer states at the bulk vertices (blue) and contracting according to the edges of the graph. In the limit of large bond dimensions, the average entanglement entropy of a boundary region $A$ is proportional to the length of a minimal cut $\gamma_A$ through the network (dashed line)~\cite{hayden2016holographic}, $S(A) \simeq S_{RT}(A)$, reproducing the Ryu-Takayanagi formula in holography.}
\label{fig:tns-and-rt}
\end{figure}

In this paper, we initiate a study of \emph{multipartite entanglement} in random tensor network models.
Our motivation is twofold:
First, recent research in quantum gravity has raised profound questions regarding the multipartite entanglement in holographic states~\cite{hayden2013holographic,balasubramanian2014multiboundary,marolf2015hot}, in particular with regards to tripartite entanglement of GHZ type~\cite{susskind2014erepr,susskind2016copenhagen}.
Answers to these questions in the context of tensor network models will likely lead to new diagnostics applicable in holography.
Second, we seek to understand the general mechanisms by which quantum information is encoded in tensor networks; an improved understanding of the entanglement structure may inform the design of tensor networks that adequately represent the physics.
While it is possible to obtain partial information from the entanglement entropy of subsystems~\cite{hayden2013holographic,balasubramanian2014multiboundary,gross2013stabilizer,bao2015holographic,cao2015holographic}, many basic questions regarding the multipartite entanglement cannot be answered from entropic data.
A striking example is that a pair of GHZ states cannot be entropically distinguished from three Bell pairs, even though their entanglement properties are vastly different~\cite{greenberger1989going}.

\ssection{Summary of results}%
Our main result is that the average amount of tripartite entanglement in random stabilizer networks is small. More precisely, for any tripartition the expected number of GHZ triples remains bounded as we take the limit of large bond dimensions (\cref{thm:main}).
This has a number of surprising consequences on the correlation and entanglement structure:
(a) The number of Bell pairs that can be extracted between two subsystems $A$ and $B$ is roughly half the mutual information $I(A:B)$ (which in turn can be read off the geometry of the network using the Ryu-Takayanagi formula);
(b) in particular, the mutual information measures quantum entanglement, proving a conjecture in~\cite{hayden2013holographic} for stabilizer tensor networks;
(c) the monogamy of the mutual information, $I(A:B)+I(A:C)\leq I(A:BC)$, established in~\cite{hayden2013holographic} for holographic entropies, thus acquires an operational interpretation as originating from the monogamy of quantum entanglement; 
(d) the \emph{tripartite information} $I_3 := I(A:B)+I(A:C)-I(A:BC)$ (i.e., the difference in the above inequality) provides a diagnostic for fourpartite entanglement; in fact, after extracting all Bell pairs we obtain a residual fourpartite entangled state with the entropies of a \emph{perfect tensor} of size $-I_3/2$~\cite{pastawski2015holographic}, strengthening the picture provided by the holographic entropy cone~\cite{bao2015holographic} (\cref{fig:entanglement}).

We establish these results based on two main technical contributions:
First, we diagnose the GHZ content by a polynomial invariant (the third moment of the partial transpose $\rho_{AB}^{T_B}$). Its average can be evaluated using a classical ferromagnetic spin model, the \emph{GHZ spin model}. For large bond dimensions, this model is in its low-temperature (ordered) phase and hence the tripartite entanglement is determined by its minimal energy configurations (\cref{fig:tripartite}).
Second, we derive a novel formula for the third moment of non-qubit stabilizer states.
It refines the results of~\cite{kueng2015qubit,zhu2015multiqubit,webb2015clifford} and we expect that it will be of similar interest in quantum information theory.
Throughout this article, we measure entropies of $p$-level systems in units of $\log_p$ bits.

\ssection{Random stabilizer networks}%
We now describe the random stabilizer network model.
Consider a connected graph with vertices $V$ and edges $E$ (parallel edges allowed).
Let $V_\partial$ denote a subset of the vertices, which we will refer to as the \emph{boundary vertices}; all other vertices are called \emph{bulk vertices} and denoted by $V_b$. 
Given a choice of bond dimensions for all edges, we define a pure quantum state by placing tensors $\ket{V_x}$ at the bulk vertices and contracting according to the edges:
\begin{equation}
\label{eq:tns}
  \ket\Psi = \left( \bigotimes_{x\in V_b} \bra{V_x} \right) \left( \bigotimes_{e\in E} \ket{e} \right)
\end{equation}
Here, $\ket{e} \propto \sum_i\ket{ii}$ denotes a normalized maximally entangled state corresponding to an edge $e$.
The state $\ket\Psi$ is a tensor network state defined on the Hilbert space corresponding to the boundary vertices $V_\partial$, and in general unnormalized.
We write $\rho=\Psi/\!\tr\Psi$ for the normalized density matrix, where $\Psi=\ket\Psi\!\!\bra\Psi$. See \cref{fig:tns-and-rt} for an illustration.

To build a stabilizer tensor network state, we choose bond dimensions of the form $D = p^N$, where $p$ is a fixed prime and $N$ some positive integer that we will later choose to be large (for simplicity of exposition, we choose all bond dimensions to be the same).
Thus the Hilbert space associated with a single vertex is of dimension $D_x = p^{N \deg(x)}$, where $\deg(x)$ denotes the degree of the vertex (i.e., the number of incident edges), and the Hilbert spaces associated with the bulk vertices has dimension $D_b = p^{N_b}$, 
where $N_b = N \sum_{x \in V_b} \deg(x)$. 
We now select each vertex tensor $V_x$ in~\eqref{eq:tns} independently and uniformly at random from the set of stabilizer states.
Thus $\Psi$ is obtained by partially projecting one stabilizer state onto another (viz., the random vertex tensors onto the maximally entangled pairs), which implies that either $\Psi$ is zero or again a stabilizer state.
In the latter case, which occurs with high probability for large $N$, we say that $\Psi$ is a \emph{random stabilizer tensor network state}.
In any tensor network state, the entanglement entropy $S(A) = -\tr\rho_A\log_p\rho_A$ of a boundary subsystem $A\subseteq V_\partial$ can always be upper bounded by
$S_{RT}(A) := N \min \lvert\gamma_A\rvert$~\cite{evenbly2011tensor},
where we minimize over all cuts $\gamma_A$ that separate the subsystem $A$ from its complement $\bar A$ in $V_\partial$ (\cref{fig:tns-and-rt}).
Formally, such a cut is defined by a subset of vertices $V_A$ that contains precisely those boundary vertices that are in $A$ 
such that the set of edges that leaves $V_A$ is $\gamma_A$.

The fundamental property of random tensor networks 
is that in the limit of large $N$ (or large $p$), this upper bound becomes saturated~\cite{hayden2016holographic}.
Thus these models reproduce the Ryu-Takayanagi formula in holography. 
More precisely, the average entanglement entropy of a boundary subsystem, conditioned on the tensor network state being nonzero, 
is given by
\begin{equation}
\label{eq:rt}
  \braket{S(A)}_{\neq0} \simeq S_{RT}(A).
\end{equation}
Here and in the following, we write $\simeq$ for equality up to order $O(1)$, independent of $N$.
The central fact used to derive this is that random stabilizer states form a projective 2-design~\cite{klappenecker2005mutually,gross2007evenly}, i.e., that their first and second moments agree with the Haar measure.
For the reader's convenience, and since the derivation in~\cite{hayden2016holographic} focused on the case of large $p$, we give a succinct derivation in~\cite{sm}\nocite{gross2006hudson}.
This result can be strengthened to show that in fact $S(A) \simeq S_{RT}(A)$ with high probability~\cite{hayden2016holographic}.

\ssection{Tripartite entanglement}%
Any pure tripartite stabilizer state $\rho_{ABC}$ is locally equivalent to a tensor product of bipartite maximally entangled states, $\ket{\Phi^+}_{AB} \propto \sum_{i=1}^p \ket{ii}$ etc., and tripartite GHZ states $\ket{\GHZ}_{ABC} \propto \sum_{i=1}^p \ket{iii}$~\cite{bravyi2006ghz,looi2011tripartite}.
That is, there exists a local unitary $U = U_A\ot U_B\ot U_C$ such that $U\rho_{ABC}U^\dagger$ is equal to
\begin{align}
\label{eq:normal form}
  (\Phi^+_{AB})^{\ot c} \ot (\Phi^+_{AC})^{\ot b} \ot (\Phi^+_{BC})^{\ot a} \ot \GHZ_{ABC}^{\ot g}
\end{align}
(we suppress local states on $A$, $B$ and $C$ which do not impact the entanglement).
The integers $a,b,c,g\geq0$ are uniquely determined; thus they meaningfully characterize the bipartite and tripartite entanglement between subsystems $A$, $B$ and $C$.
Now we can state our main result:

\begin{figure}
\raisebox{1.2cm}{(a)}
\raisebox{-0.15cm}{\includegraphics[height=3.2cm]{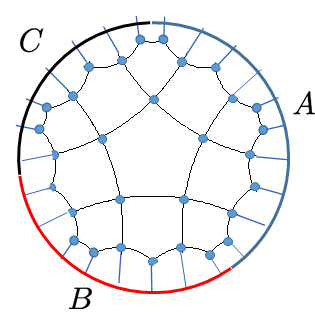}}\quad
\raisebox{1.2cm}{(b)~~}
\includegraphics[height=3cm]{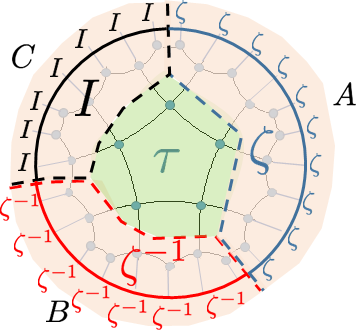}
\caption{\textbf{Tripartite entanglement and the GHZ spin model.} (a)~Tripartition of the boundary. (b)~Illustration of the spin model (with boundary conditions and minimal energy configuration) used to evaluate the GHZ content of a random stabilizer tensor network state.}
\label{fig:tripartite}
\end{figure}

\begin{thm}[Tripartite entanglement in random stabilizer networks]
\label{thm:main}
  Let $A$, $B$, $C$ denote a tripartition of the boundary (\cref{fig:tripartite}~(a)), and $p\equiv2\pmod3$.
  Then the expected number of GHZ states in a random stabilizer network is of order $O(1)$ in the limit of large $N$. 

  Explicitly, we have the following bound in terms of the geometry of the tensor network:
  \[ \braket{g}_{\neq 0} \leq \#_b \log_p(p+1) + \log_p(\#_A\#_B\#_C) + 4\delta, \]
  with $\#_A$ the number of minimal cuts for $A$, etc.,
  $\#_b$ the maximal number of components of any subgraph obtained by removing minimal cuts for $A$, $B$ $C$~\footnote{In the language of~\cite{pastawski2015holographic}, $\#_b$ is the number of \emph{multipartite residual regions}.},
  and $\delta=(2p+2)^{V_b}/p^N$.
\end{thm}
In most cases of interest, the minimal cuts are unique and there remains a single connected component after their removal, so that $\braket{g}_{\neq0} \leq \log_p(p+1) + 4\delta$~\footnote{Note that $S(A)+S(B)+S(C)=2(a+b+c)+3g$. It follows that if the sum of local entropies is odd then, necessarily, $g>0$. This is all that can be said about the tripartite entanglement from the knowledge of the entropies alone, and it justifies that the upper bound in \cref{thm:main} is never smaller than $\log_p(p+1)\geq1$, 
even when the minimal cuts are unique.}.
We note that Markov's inequality implies that the number of GHZ triples in fact remains bounded with high probability.
\Cref{thm:main} vastly generalizes the bound in \cite{smith2006typical}, which can be obtained as the special case for a graph with a \emph{single} bulk vertex.

In general, the mutual information is sensitive to both classical and quantum correlations.
For a general stabilizer state of the form~\eqref{eq:normal form}, $I(A:B) = 2c+g$, where $c$ is the number of maximally entangled pairs and $g$ the number of GHZ triples (whose reduced state on AB is a classically correlated state). 
In random stabilizer networks, however, \cref{thm:main} shows that $\braket{g}_{\neq0}$ is bounded.
Thus the average number of maximally entangled pairs that can be extracted between $A$ and $B$ is roughly one half the mutual information, $I(A:B)/2\simeq c$, which in turn can be estimated from the geometry of the tensor network by using the Ryu-Takayanagi formula~\eqref{eq:rt}.
In particular, bipartite correlations between any two boundary subsystems are dominated by quantum entanglement and determined rigidly by the geometry of the tensor network, confirming a property that is also suspected to hold in holography~\cite{hayden2013holographic} (\cref{fig:entanglement}~(a)).


\begin{figure}
\raisebox{1.2cm}{(a)}
\raisebox{0.2cm}{\includegraphics[height=1.8cm]{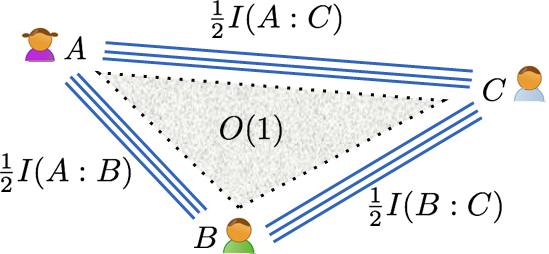}} \quad
\raisebox{1.2cm}{(b)}
\includegraphics[height=2.5cm]{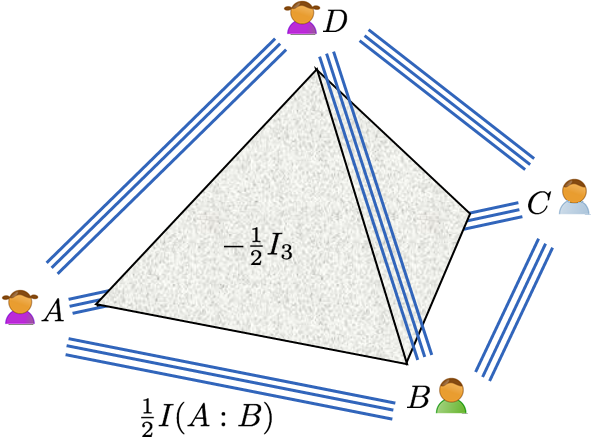}
\caption{\textbf{Multipartite entanglement structure.}
(a)~For any tripartition, there is only a bounded number of GHZ triples (dashed triangle) and hence the entanglement is dominated by bipartite maximal entanglement (blue lines).
(b)~For four (and more) parties, we can likewise extract maximally entangled pairs between any two parties (blue lines).
The residual state has approximately the entropies of a perfect tensor (tetrahedron).
This decomposition is in one-to-one correspondence with the extreme rays of the holographic entropy cone~\cite{bao2015holographic}.}
\label{fig:entanglement}
\end{figure}

\ssection{Higher-partite entanglement}%
\Cref{thm:main} has a number of remarkable consequences for the entanglement structure for four and more subsystems.
We first consider the extraction of bipartite entanglement.
Consider a random stabilizer tensor network state whose boundary is partitioned into $k$ subsystems $A_1, A_2, \dots, A_k$.
Applying the preceding discussion to $A=A_i$, $B=A_j$ and $C=\overline{A_iA_j}$ their complement, we find that the average number of maximally entangled pairs that can be extracted between any two subsystems $A_i$ and $A_j$ is $t_{ij}\simeq I(A_i:A_j)/2$.
The extraction process is implemented by local unitaries $U_i\ot U_j$; it leaves all other mutual informations invariant and does not introduce new GHZ triples.
We can therefore repeat the process and extract maximally entangled pairs between any pair of subsystems $A_i$ and $A_j$, until we obtain a residual state $\tilde\rho_{A_1\dots A_k}$ whose bipartite mutual informations $I(A_i:A_j)$ are all of order $O(1)$.

We now specialize the preceding discussion to a fourpartite system ($k=4$).
Here, the vanishing of the pairwise mutual informations implies that the entropies of the residual state will have the following simple form:
$S(A_i)\simeq \frac12 S(A_iA_j) \simeq m$ for all $i\neq j$, where $m\geq0$ is some integer~\cite{bao2015holographic}.
Ignoring the order-one corrections, stabilizer states with such entropies are fourpartite \emph{perfect tensors}.
These are tensors that are unitaries from any pair of subsystems to the complement, a crucial property used in the explicit construction of holographic codes~\cite{pastawski2015holographic,yang2015bidirectional}.
Significantly, it is possible to determine $m$ from the entropies of the original state, or, more specifically, from its \emph{tripartite information} $I_3 := I(A_1:A_2)+I(A_1:A_3)-I(A_1:A_2A_3)$, which is invariant under the extraction of the maximally entangled pairs (it also does not depend on the choice of $A_1,A_2,A_3$).
In short, we have established the following result:

\begin{thm}[Fourpartite entanglement in random stabilizer networks]
\label{thm:fourpartite}
  Let $A_1,\dots,A_4$ denote a partition of the boundary into four subsystems.
  Then the random stabilizer network state is locally equivalent to
  \begin{equation}
  \label{eq:fourpartite}
    \bigotimes_{i\neq j} (\Phi^+_{A_iA_j})^{t_{ij}} \ot \tilde\rho_{A_1A_2A_3A_4},
  \end{equation}
  In the limit of large $N$, on average $t_{ij} \simeq \frac12I(A_i:A_j)$ and the residual state $\tilde\rho$ has approximately the entropies of a perfect tensor of size $-I_3/2$ (that is, $S(A_i)\simeq S(A_iA_j)/2 \simeq -I_3/2$). 
\end{thm}

Our result provides a new interpretation of the tripartite information $I_3$ for random stabilizer networks -- namely, as a measure of the entropy of the residual, genuinely fourpartite entangled state $\tilde\rho$.
Since entropies are always nonnegative, it follows that $I_3\lesssim0$; equivalently, the mutual information is monogamous, $I(A:B)+I(A:C)\lesssim I(A:BC)$, as was proved for holographic entropies in~\cite{hayden2013holographic}.
This can also be seen by observing that, in our setting, one half the mutual information is an entanglement measure; it is up to $O(1)$ corrections equal to, e.g., the squashed entanglement $E_{sq}$~\cite{christandl2004squashed}; 
therefore the monogamy of the mutual information also follows as a direct consequence of the monogamy of the latter.

It is also interesting to compare \cref{thm:fourpartite} with the classification of fourpartite holographic entropies in~\cite{bao2015holographic}.
We find that there is a one-to-one correspondence between the building blocks of fourpartite entanglement in~\eqref{eq:fourpartite} and the extreme rays of the fourpartite holographic entropy cone defined in~\cite{bao2015holographic}.
That is, the entropies of a four-partite holographic state can always be reproduced by states of the form~\eqref{eq:fourpartite} (up to rescaling).
\Cref{thm:fourpartite} elevates this result from the level of entropies to the level of quantum states for random stabilizer networks.
It is natural to ask if this correspondence can be extended to a higher number of parties, where the phase space of holographic entropies becomes significantly more complicated.

Lastly, we note that while many important many-body states are stabilizers (e.g., ground states of commuting Pauli Hamiltonians, such as the toric code and several fracton models~\cite{haah2013commuting}, as well as states in certain Chern-Simons theories~\cite{salton2017entanglement}), most states are far from this ensemble.
It is an interesting open problem to generalize our results to other scenarios.

\ssection{Method: The GHZ spin model}%
We now sketch the proof of \cref{thm:main}.
Previous works such as~\cite{smith2006typical} have calculated the GHZ content of multiqubit stabilizer states by using the algebraic formula from \cite{bravyi2006ghz} in terms of dimensions of co-local stabilizer subgroups.
Here, we proceed differently.
The idea is to use the partial transpose $\rho_{AB}^{T_B}$ of the reduced state, which is sensitive to bipartite entanglement.
A short calculation using~\eqref{eq:normal form} shows that $\tr (\rho_{AB}^{T_B})^3 = p^{-2(a+b+c+g)}$.
Thus the number of GHZ states contained in a tripartite stabilizer state can be computed as
\begin{equation}
\label{eq:ghz no}
g = S(A) + S(B) + S(C) + \log_p \tr (\rho_{AB}^{T_B})^3.
\end{equation}
In a random stabilizer network, we can upper-bound $S(A)\leq S_{RT}(A)$ etc., and we know from the preceding section that this bound is not too lose.
The main challenge is to upper-bound the expectation value $\braket{\tr (\Psi_{AB}^{T_B})^3}$, which is a third moment in the \emph{unnormalized} random tensor network state~\eqref{eq:tns}.
In general, it is well known that a mixed quantum state $\rho_{AB}$ has bipartite entanglement if $\rho_{AB}^{T_B}$ has negative eigenvalues, hence, moments of $\rho_{AB}^{T_B}$ should contain information about the multipartite entanglement of the global pure state.%
\footnote{Note that $\tr \rho_{AB}^{T_B} = 1$, and $\tr (\rho_{AB}^{T_B})^2$ is the purity of system $A$, so is an entropic measure. The third moment is the smallest moment containing nontrivial information about the multipartite entanglement.}
This connection is particularly sharp for stabilizer states through~\cref{eq:ghz no}, but we expect similar calculations to be informative for other ensembles of quantum states.

We start with the multiqubit case ($p=2$).
Only in this case, we can use the recent result that multiqubit stabilizers are projective 3-designs~\cite{zhu2015multiqubit,webb2015clifford}.
Thus we have that for each vertex tensor $\braket{\ket{V_x}\!\!\bra{V_x}^{\ot3}}=\sum_{\pi\in S_3} R_x(\pi) /D_x(D_x+1)(D_x+2)$, where we sum over all permutations $\pi\in S_3$ and write $R_x(\pi)$ for the corresponding permutation operator acting on three copies of the vertex Hilbert space.
Using the analogous notation, we find that $\tr (\Psi_{AB}^{T_B})^3 = \tr \Psi^{\ot 3} R_A(\zeta) R_B(\zeta^{-1})$, where $\zeta$ is the cyclic permutation that sends $1\mapsto2\mapsto3$.
A careful calculation then reveals that
\begin{equation}
\label{eq:careful qubits}
  \Braket{\tr (\Psi_{AB}^{T_B})^3}
\leq 2^{-3N_b} \sum_{\{\pi_x\}} 2^{-N \sum_{\braket{xy}} d(\pi_x, \pi_y)}
\end{equation}
where the sum is over all choices of permutations $\pi_x \in S_3$, subject to the boundary conditions $\pi_x=\zeta$ for $x\in A$, $\pi_x=\zeta^{-1}$ in $B$, and $\pi_x=1$ in $C$; the sum in the exponent is over all edges, and we define $d(\pi_x,\pi_y)$ as the minimal number of transpositions required to go from one permutation to the other.
We can interpret the right-hand side of~\eqref{eq:careful qubits} as the partition sum of a ferromagnetic spin model with permutation degrees of freedom at each vertex at inverse temperature $\log N$ (\cref{fig:tripartite}~(b)). 

For large $N$, we are in the low-temperature (ordered) phase and the partition function is dominated by the minimal energy configuration:
\[ \sum_{\{\pi_x\}} 2^{-N \sum_{\braket{xy}} d(\pi_x, \pi_y)} \leq 2^{-N E_0} \bigl( \# + \delta \bigr), \]
where $E_0$ denotes the minimal energy, $\#$ the number of minimal energy configurations and $\delta=6^{V_b}/2^N$.
Now consider an arbitrary configuration $\{s_x\}$, minimal or not.
If we denote by $V_A$ the $\zeta$-domain then the boundary conditions ensure that 
$V_A$ is a cut separating $A$ from $BC$.
While this cut is not necessarily minimal, we always have that $N \lvert\partial V_A\rvert \geq S_{RT}(A)$, where $\lvert\partial V_A\rvert$ denotes the number of edges that leaves $V_A$.
Likewise, the $\zeta^{-1}$-domain $V_B$ is a cut for $B$ and the identity domain $V_C$ is a cut for $C$, so that $N \lvert\partial V_B\rvert\geq S_{RT}(B)$ and $N \lvert\partial V_C\rvert\geq S_{RT}(C)$.
For each edge leaving $V_A$, the energy cost is at least $1$, and it is $2$ if the edge enters one of the domains $V_B$ or $V_C$ (since $1,\zeta,\zeta^{-1}$ are even permutations). 
Thus the energy cost of an arbitrary configuration $\{s_x\}$ can be lower bounded by
$N E[\{s_x\}] \geq S_{RT}(A) + S_{RT}(B) + S_{RT}(C)$,
with equality if and only if all three domains $V_A$, $V_B$, and $V_C$ are disjoint minimal cuts and if each connected component of the remaining bulk vertices
is assigned a transposition.
This can always be achieved, so 
\[ E_0 = \bigl( S_{RT}(A) + S_{RT}(B) + S_{RT}(C) \bigr)/N, \]
with degeneracy $\#\leq 3^{\#_b} \#_A\#_B\#_C$, since there are three possible transpositions to choose from for each component (\cref{fig:tripartite}~(b)).
Combining these estimates with~\eqref{eq:ghz no} and using basic properties of the trace, we obtain \cref{thm:main} for qubits.

\smallskip

For $p\neq2$, the stabilizer states no longer form a projective 3-design.
To generalize our preceding argument, we derive a new formula for the third moment of a random stabilizer state $\ket V$ in $(\CC^p)^{\ot n}$, where $p\equiv2\pmod3$ and $n\geq3$~\cite{sm}:
\begin{equation}
\label{eq:third}
  \bigl\langle \ket V\!\!\bra V^{\ot 3} \bigr\rangle
  = \frac 1 {p^n(p^n + 1)(p^n+p)} \sum_T R(T)
\end{equation}
The sum is over the group $G_3(p)$ of orthogonal and doubly stochastic $3\times3$-matrices with entries in $\FF_p$; $R(T)$ is the corresponding operator defined on $(\CC^p)^{\ot 3n}$ by $R(T)=r(T)^{\ot n}$, $r(T) \ket{\vec q} = \ket{T\vec q}$ for $\vec q \in \FF_p^3$.
For qubits, $G_3(p)$ is equal to the permutation group; in general, it contains the latter as a proper subgroup.
In contrast to previous results, which compute the frame potential of stabilizer states~\cite{kueng2015qubit,zhu2015multiqubit,webb2015clifford}, our formula can be used to evaluate arbitrary third moments; we expect that~\eqref{eq:third} will be of independent interest in quantum information theory.

Just like in the case of qubits, the operators $R(T)$ act as a tensor product with respect to the $n$ copies of the single-particle replica Hilbert space $(\CC^p)^{\ot 3}$.
This is the central property that allows us to adapt the argument given above for qubits to obtain a classical ferromagnetic spin model with $G_3(p)$ degrees of freedom.
\Cref{thm:main} follows as above by an analysis of the low-temperature behavior of this model.
See~\cite{sm} for the technical details.

\ssection{Discussion and outlook}%
We have initiated a comprehensive study of multipartite entanglement in tensor network models of holography.
Our results suggest several avenues for further investigation:
First, it would be of mathematical interest to extend our analysis and establish sharp deviation bounds as in~\cite{smith2006typical}.
Second, tensor networks can also be used to define bulk-boundary mappings, or `holographic codes'~\cite{pastawski2015holographic,yang2015bidirectional,bao2015holographic}.
In this case, the entanglement entropies of code states obtain a bulk correction, in agreement with the expectations of AdS/CFT~\cite{faulkner2013quantum}, and it is natural to ask in which way the multipartite entanglement of typical code states is determined by the bulk~\footnote{In light of our approach, it is natural to conjecture that $S(\Phi_{V_A}) + S(\Phi_{V_B}) + S(\Phi_{V_C}) + \log_p \tr [\Phi^{\ot 3} R_{V_A}(\zeta) R_{V_B}(\zeta^{-1}) R_{V'}(\tau)]$, where $\Phi$ is the inserted bulk state, will play a significant role.}.
Third, diagnostics such as moments of the partial transpose considered in this paper may provide a path towards generalizing our results to non-stabilizer states and lead to a more refined understanding of multipartite entanglement in the AdS/CFT correspondence.

Random tensor networks have been a crucial source of inspiration for recent developments in the information theory of quantum gravity, in part due to complete analytical control over their bipartite entanglement structure.
Some important examples include entanglement wedge reconstruction~\cite{dong2016reconstruction} and the recent progress on understanding the black hole information paradox~\cite{penington2019entanglement}.
In some cases, the connections go beyond mere inspiration, for instance, fixed-area states in quantum gravity mimic the entanglement properties of random tensor networks~\cite{dong2019flat}.
The stabilizer random tensor network model presented in this paper shares these very same features, while in addition allowing precise analytical access to the multipartite entanglement structure.

\ssection{Acknowledgements}%
It is a pleasure to thank David Gross, Patrick Hayden, Debbie Leung, Xiao-Liang Qi, Lenny Suss\-kind, Zhao Yang, Huangjun Zhu for inspiring discussions.
SN acknowledges support of Stanford Graduate Fellowship.
MW gratefully acknowledges support from FQXI, the Simons Foundation, the DoD Multidisciplinary University Research Initiative (MURI), and an NWO Veni grant (no.~680-47-459).
\bibliography{randomstabs}
\clearpage

\appendix
\begin{widetext}
\begin{center}
    {\Large Supplemental Material}
\end{center}
\section{Quantization of the trace}\label{app:trace}
In~\cite[App.~F]{hayden2016holographic}, it was shown that if $\ket{\phi_A}\in(\CC^p)^{\ot a}$ and $\ket{\psi_{AB}}\in(\CC^p)^{\ot(a+b)}$ are stabilizer states, with corresponding stabilizer groups $G$ and $H$, then the projection $\ket\Psi_B = \braket{\phi_A | \psi_{AB}}$, if nonzero, is given by
\[ \Psi_B = \frac {\lvert K\rvert} {\lvert H \rvert} \frac 1 {\lvert L \rvert} \sum_{g_B\in L} g_B, \]
where $K$ some subgroup of $G\times H$ and $L$ a commutative subgroup of the corresponding Weyl-Heisenberg group, implying that $\Psi_B$ is again a stabilizer state.
The order of both $K$ and $H$ is a power of $p$, so that $\tr\Psi_B = \lvert K\rvert/\lvert H \rvert$ is necessarily quantized in powers of $p$.
Moreover, $L$ was defined in \cite{hayden2016holographic} as the homomorphic image of $K$, so that $\lvert K\rvert\geq\lvert L\rvert$, and hence $\tr\Psi_B\geq\lvert L\rvert/\lvert H\rvert=p^b/p^{a+b}$, since $\lvert L\rvert=p^b$ and $\lvert H\rvert=p^{a+b}$.
Thus we find that $\tr\Psi_B=p^k/p^a$, where $k=0,\dots,a$.

Applied to the tensor network state $\ket\Psi$ defined in~\eqref{eq:tns}, where the vertex tensors $\ket{V_x}$ are stabilizer states, we note that $\ket\Psi$ is obtained by projecting the collection of Bell pairs onto the tensor product $\bigotimes_x \ket{V_x}$, which is a stabilizer state in $(\CC^p)^{\ot N_b}$.
Thus we obtain that $\ket\Psi$ is either zero or again a stabilizer state, with trace $\tr\Psi=p^k/p^{N_b}$, where $k=0,\dots,N_b$.

\section{Proof of the Ryu-Takayanagi formula}\label{app:rt}
We give a succinct derivation of the lower bound on the average entanglement entropy.
The central fact that we will use is that random stabilizer states form a projective 2-design~\cite{klappenecker2005mutually,gross2007evenly}.
Thus their first and second moments agree with the Haar measure; we have that $\braket{\ket V\!\!\bra V_x} = I / D_x$ and $\braket{\ket V\!\!\bra V_x^{\ot 2}} = (I + F_x)/D_x(D_x+1)$, where $I$ denotes identity operators and $F_x$ the swap operator on two copies of the Hilbert space of vertex $x$.
The former readily gives
\[ \Braket{\tr\Psi} 
= p^{-N_b}, \]
and using the latter it can be quickly calculated that
\begin{align*}
  \Braket{\tr\Psi_A^2}
= \tr\Braket{\Psi^{\ot2}} F_A
= \frac 1 {\prod_{x\in V_b} D_x(D_x+1)} \tr \left[ \left(\prod_e \ket e\!\!\bra e^{\ot 2}\right) \left( \prod_{x\in V_b} (I + F_x) \right) F_A\right]
\leq p^{-2 N_b} \sum_{V_A\cap V_\partial=A} p^{-N \lvert\partial V_A\rvert)},
\end{align*}
where we have used that each $\ket e$ is a maximally entangled state of rank $D=p^N$; we recall that $\lvert\partial V_A\rvert$ denotes the number of edges that leaves $V_A$.
The right-hand side sum is over all cuts $V_A$ between $A$ and $\bar A$, as explained in the main text.
It is plain that the sum will be dominated by the minimal cuts, 
as all other cuts are suppressed by a factor $1/p^N$ or more.
Thus,
\begin{equation}
\label{eq:second moment upper bound}
  \Braket{\tr\Psi_A^2} \leq p^{-2 N_b} p^{-S_{RT}(A)} \Bigl( \#_A + \varepsilon \Bigr),
\end{equation}
where $\#_A$ is the number of minimal cuts and $\varepsilon := 2^V / p^N$.
This calculation has two important consequences:

First, for $A=\emptyset$ we have that $\Psi_A=\tr\Psi$, so the above can be used to bound the fluctuations of the trace of the unnormalized tensor network state~\eqref{eq:tns}.
Here, $\#_A=1$ as long as each connected component of the graph contains at least one boundary vertex (so in particular if the graph is connected), so that $\braket{(\tr\Psi)^2} \leq p^{-2 N_b} (1 + \varepsilon)$.
From \cref{app:trace} we know that if $\Psi\neq0$ then $\tr\Psi=p^k/p^{N_b}$ for some integer $k=0,1,\dots,N_b$.
Let us write $q_k$ for the probability that $\tr\Psi=p^k/p^{N_b}$; we are interested in bounding $q_0$.
Then we obtain the following two relations from the first and second moment of $\tr\Psi$ computed above:
\[ \sum_{k=0}^{n_V} q_k p^k = 1, \quad \sum_{k=0}^{n_V} q_k p^{2k}\leq1+\varepsilon. \]
It follows that $1+\varepsilon\geq q_0 + p\sum_{k=1}^{n_V} q_k p^{2k-1}\geq q_0 + p\sum_{k=1}^{n_V} q_k p^k=q_0 + p(1-q_0)=(1-p)q_0+p$ and hence that $q_0\geq1-\frac\varepsilon{p-1}\geq1-\varepsilon$.
In other words,
\begin{equation}
\label{eq:prob trace nz}
  \Pr(\Psi\neq0) \geq \Pr(\tr\Psi=p^{-N_b}) = q_0 \geq 1-\varepsilon.
\end{equation}
Thus we do not only find that $\Psi\neq0$, but in fact that the trace is equal to its expected and minimal value with high probability as $N$ or $p$ becomes large.

Second, recall that the entanglement entropy can always be lower-bounded by the R\'enyi-2 entropy $S_2(A) = -\log_p\tr\rho_A^2$.
For stabilizer states we in fact have equality, as their entanglement spectra are flat, and thus
\[ \braket{S(A)}_{\neq0}
= 2\braket{\log_p\tr\Psi}_{\neq0} - \braket{\log_p\tr\Psi_A^2}_{\neq0}, \]
where we write $\braket{f}_{\neq0}$ for the expectation value of an observable $f$ conditioned tensor network state being nonzero ($\Psi\neq0$).
Using the fact that $\tr\Psi\geq p^{-N_b}$ if $\Psi\neq0$, Jensen's inequality for the (concave) logarithm, and $\braket{\tr\Psi_A^2} = \braket{\tr\Psi_A^2}_{\neq0} \Pr(\Psi\neq0)$, we can bound this as
\begin{align*}
  \braket{S(A)}_{\neq 0}
\geq -2 N_b - \log_p \braket{\tr\Psi_A^2} + \log_p (1-\varepsilon)
\geq S_{RT}(A) - \log_p ( \#_A + \varepsilon ) + \log_p (1-\varepsilon).
\end{align*}
where we have plugged in the upper bound~\eqref{eq:second moment upper bound} to obtain the second inequality.
Since $\varepsilon$ is arbitrarily small for large enough $N$ or $p$, we obtain that
\[ \braket{S(A)}_{\neq0} \geq S_{RT}(A) - \log_p \#_A - 4\varepsilon, \]
where $\#_A$ is the number of minimal cuts.
Thus the expected entanglement entropy of a boundary subsystem in a random stabilizer network is indeed close to saturating the Ryu-Takayanagi formula.

\section{Third moment of stabilizer states}\label{app:third}
In this section we give a detail proof of our formula~\eqref{eq:third} for the third moment of a random pure stabilizer state in $(\CC^p)^{\ot n}$ with local dimension $p\equiv2\pmod3$ and $n\geq3$.

Let $T$ be an invertible $3\times3$-matrix with entries in $\FF_p$.
The set of all such matrices is the general linear group $\GL_3(p)$.
We consider the representation~$r(T)$ of $\GL_3(p)$ on $(\CC^p)^{\ot 3}$, given by $r(T)\ket{\vec q}=\ket{T\vec q}$, and its $n$-fold tensor power $R(T) := r(T)^{\ot n}$ on $((\CC^p)^{\ot 3})^{\ot n} \cong (\CC^p)^{3n}$.
We note that $R(T)$ is represented by real orthogonal matrices in the computational basis (in fact, by a permutation matrix).

We say that $T$ is \emph{orthogonal} if $T T^t = T^t T = I$, and we call $T$ \emph{doubly stochastic} if its row sums and column sums are equal to $1\pmod p$.
Let $G_3(p)$ denote the group of orthogonal and doubly stochastic $3\times 3$-matrices with entries in $\FF_p$.
We note that a row-stochastic (or column-stochastic) orthogonal matrix is automatically doubly stochastic.
It is plain that $G_3(p)$ contains the group of permutation matrices as a subgroup, which we will identify with the permutation group $S_3$.
Moreover, for any permutation matrix $\pi$, $R(\pi)$ agrees with the usual permutation action of $S_3$ on $((\CC^p)^{\ot n})^{\ot3}$.
We will give an explicit description of $G_3(p)$ in \cref{eq:concrete elements,eq:explicit} below.

For qubits, $p=2$, it is easy to see in fact any orthogonal and doubly stochastic matrix is a permutation matrix and hence $G_3(p) = S_3$.
Thus~\eqref{eq:third} is reduces to $\braket{\ket V\!\!\bra V^{\ot 3}} = \sum_{\pi\in S_3} R(\pi) / 2^n(2^n + 1)(2^n+2)$, which follows directly from the recent result that multiqubit stabilizer states form a projective 3-design~\cite{kueng2015qubit,zhu2015multiqubit,webb2015clifford}.
For odd primes $p$, however, this is no longer the case and we have to develop new methods.

The set of stabilizer states $\Stab(n,p)$ on $(\CC^p)^{\ot n}$ is a single orbit under the Clifford group $\Cliff(n,p)$.
In particular, the third moment $\braket{\ket V\!\!\bra V^{\ot 3}}$ is an operator that commutes with $U^{\ot 3}$ for any Clifford unitary $U\in\Cliff(n,p)$, i.e., an element of the \emph{commutant} of $\Cliff(n,p)^{\ot 3}$.
For qubits, this commutant is generated by the permutation action $R(\pi)$ for $\pi\in S_3$ (indeed, this implies that multiqubit stabilizer states form a 3-design).
We will show that the analogous statement holds true for general $p\equiv2\pmod 3$ if we consider the larger group of orthogonal and doubly stochastic matrices in $G_3(p)$ (\cref{thm:commutant}); this result will in turn imply~\eqref{eq:third} at once:

\begin{thm}
\label{thm:commutant}
  Let $p\equiv2\pmod3$ be a prime and $n\geq3$.
  Then the operators $R(T)$ for $T \in G_3(p)$ are $2p+2$ linearly independent operators that span the commutant of $\Cliff(n,p)^{\ot 3}$.
\end{thm}

To prove \cref{thm:commutant} we need some intermediate results which are of independent interest.
We start by analyzing the phase space picture for odd $p$.
Any Clifford unitary $U \in \Cliff(n,p)$ can be parametrized by a symplectic matrix $S \in \Sp(2n,p)$ and a vector $b \in \FF_p^{2n}$~\cite{gross2006hudson}, such that $U A(x) U^\dagger = A(Sx+b)$ for any phase space point operator $A(x)$, $x \in \FF_p^{2n}$.
That is, the conjugation action of $U$ corresponds to the affine action $x\mapsto Ax+b$ on phase space.
Now consider the three-fold replica Hilbert space $((\CC^p)^{\ot n})^{\ot3}$.
The corresponding classical phase space $\FF_p^{2(3n)}$ can be identified with $\FF_p^{2n} \ot \FF_p^3$, where the second factor corresponds to the three-fold replica (the tensor product is over the finite field $\FF_p$).
From this perspective, $U^{\ot 3}$ is again a Clifford unitary, corresponding to the affine transformation
\begin{equation}
\label{eq:three-fold classical}
  x\mapsto \bigl(S \ot I_3\bigr) x + b \ot \left(\begin{smallmatrix}1\\1\\1\end{smallmatrix}\right)
\end{equation}
On the other hand, a direct calculation shows that $R(T) A(q,p) R(T)^\dagger = A((I_n \ot T)q, (I_n\ot T^{-t})q)$ for any $T \in \GL_3(p)$ and $(q,p) \in \FF_p^{2(3n)}$.
Thus $R(T)$ is a Clifford unitary in $\Cliff(3n,p)$.
If we restrict to orthogonal matrices, then $T^{-t} = T$, and hence $R(T)$ corresponds to the affine transformation
\begin{equation}
\label{eq:T classical}
  x \mapsto \bigl(I_{2n} \ot T\bigr)(x)
\end{equation}
If $T$ is row stochastic then $T \left(\begin{smallmatrix}1\\1\\1\end{smallmatrix}\right) = \left(\begin{smallmatrix}1\\1\\1\end{smallmatrix}\right)$ and hence the phase space transformations~\eqref{eq:three-fold classical} and~\eqref{eq:T classical} commute with each other.
(Conversely, it is easy to see that these conditions are also necessary for the two transformations to commute.)
We can now establish the following lemma:

\begin{lem}
\label{lem:commuting}
  For any odd prime $p$ and $n$, the operators $R(T)$ for $T \in G_3(p)$ commute with any $U^{\ot 3}$ for $U \in \Cliff(n,p)$.
\end{lem}
\begin{proof}
  We have just seen that the phase space transformations corresponding to $U^{\ot 3}$ and $R(T)$ commute with each other.
  That is, $U^{\ot3} R(T) A(x) R(T)^\dagger (U^\dagger)^{\ot3} = R(T) U^{\ot3} A(x) (U^\dagger)^{\ot3} R(T)^\dagger$ for any phase space point operator $A(x)$, which implies that $U^{\ot 3} R(T) = \gamma R(T) U^{\ot 3}$ for some global phase $\gamma\in U(1)$.

  To fix the phase, note that $R(T) \ket{x}^{\ot 3} = \ket{x}^3$ by row stochasticity.
  Now consider some nonzero matrix element $\braket{x|U|y}\neq0$, where $x\in\FF_p^n$ and $y\in\FF_p^n$.
  Then, $\bra x^{\ot 3} U^{\ot 3} R(T) \ket y^{\ot 3} = \braket{x|U|y}^3 = \bra x^{\ot 3} R(T) U^{\ot 3} \ket y^{\ot 3}$, which shows that $\gamma=1$.
\end{proof}

The dimension of the commutant of $\Cliff(n,p)^{\ot 3}$ is known as the \emph{third frame potential} of the Clifford group, denoted $\Phi_3$ in~\cite{zhu2015multiqubit}.
It can be evaluated by counting the orbits of the diagonal action of the symplectic group on \emph{two} copies of the phase space.
The result is that $\Phi_3 = 2p + 2$ for $n\geq2$~\cite[eq.~(9)]{zhu2015multiqubit}.
Thus in order to establish \cref{thm:commutant} it suffices to exhibit $2p+2$ linearly independent operators in the commutant of $\Cliff(n,p)^{\ot 3}$.

\begin{lem}
\label{lem:indep}
  If $n\geq3$ then operators $R(T)$ are linearly independent.
\end{lem}
\begin{proof}
  Let $e_1$, $e_2$, $e_3$ denote the first three standard basis vectors of $\FF_p^n$.
  Then
  $R(T) \ket{e_1,e_2,e_3} = \ket{t_1,t_2,t_3} =: \ket{T},$
  where $t_i = T_{i,1} e_1 + T_{i,2} e_2 + T_{i,3} e_3 \in \FF_p^n$ is equal to the $i$-th row of $T$, extended suitably by zeros.
  Clearly, $\braket{T|T'}=0$ for $T\neq T'$, and hence the operators $R(T)$ are indeed linearly independent.
\end{proof}

In view of the preceding, the $R(T)$ for $T \in G_3(p)$ are linearly independent operators in the commutant of $\Cliff(n,p)^{\ot3}$, which is of dimension $2p+2$.
We now explicitly construct $2p+2$ distinct -- and therefore \emph{all} -- matrices $T \in G_3(p)$.
For this, consider the following numbers in $\FF_p$
\[ a_m = \frac {1+m} {1+m+m^2}, \quad b_m = \frac {-m} {1+m+m^2}, \quad c_m = \frac {m+m^2} {1+m+m^2} \qquad (m=0,\dots,p-1), \]
where all arithmetic is modulo $p$.
The third cyclotomic polynomial $1+X+X^2$ is irreducible if (and only if) $p\equiv2\pmod3$, hence the division by $1+m+m^2$ is a valid operation.
Consider, furthermore,
\[ a_p = b_p = 0, \quad c_p = 1. \]
It is easily verified that the $p+1$ triples $(a_m,b_m,c_m)$ for $m=0,\dots,p$ are all distinct.
Now recall that the permutation group can be decompose into the even and odd permutations, $S_3 = \{1, \zeta, \zeta^{-1}\} \cup \{\tau_{12}, \tau_{13}, \tau_{23}\}$, where $\zeta$ is the cyclic permutation that sends $1\mapsto2\mapsto3$ and $\tau_{ij}$ the transposition that interchanges $i\leftrightarrow j$.
We correspondingly define an \emph{even} and an \emph{odd} matrix for each $m=0,\dots,p$:
\begin{equation}
  \label{eq:concrete elements}
  T_{m,\text{even}} = a_m + b_m \zeta + c_m \zeta^{-1} = \begin{pmatrix} a_m & c_m & b_m \\ b_m & a_m & c_m \\ c_m & b_m & a_m\end{pmatrix},
\qquad T_{m,\text{odd}} = a_m \tau_{12} + b_m \tau_{13} + c_m \tau_{23} = \begin{pmatrix} c_m & a_m & b_m \\ a_m & b_m & c_m \\ b_m & c_m & a_m\end{pmatrix}
\end{equation}
It can be seen by direct inspection that the matrices $T_{m,\text{even}}$ and $T_{m,\text{odd}}$ thus defined are orthogonal and doubly stochastic 
(indeed, we have that $a_m + b_m + c_m = a_m^2 + b_m^2 + c_m^2 = 1$, while $a_m b_m + a_m c_m + b_m c_m = 0$).
This concludes the proof of \cref{thm:commutant}.

\smallskip

The preceding discussion shows that, for $p\equiv2\pmod3$ and $n\geq3$, we can write
\begin{equation}
\label{eq:explicit}
  G_3(p) = \{ T_{m,\text{even}} : m=0,\dots,p \} \cup \{ T_{m,\text{odd}} : m=0,\dots,p \}.
\end{equation}
It is plain that our notion of even and odd elements in $G_3(p)$ specializes to the definition for the subgroup of permutation matrices $S_3 \subseteq G_3(p)$.
Moreover, just as for the permutation group, the product $T T'$ of any two elements in $G_3(p)$ is even if and only if $T$ and $T'$ are both even or both odd; in particular, the even elements $\{T_{m,\text{even}}\}$ form a subgroup of $G_3(p)$.
We now compute the trace of each representation matrix $R(T)$. Since $R(T)$ acts by permuting the computational basis vectors of $(\CC^p)^{\ot3n}$, its trace is equal to the number of fixed points, hence
$\tr R(T) = p^{\dim \ker (I_n \ot T - I_{3n})} = p^{n \ker (T - I)}$.
A direct calculation shows that
\begin{align*}
  &\tr R(T_{0,\text{even}}) = \tr I = p^{3n}, & \\
  &\tr R(T_{m,\text{even}}) = p^{n} & (m=1,\dots,p), \\
  &\tr R(T_{m,\text{odd}}) = p^{2n} & (m=0,\dots,p).
\end{align*}
In particular,
\begin{equation}
\label{eq:sum of traces}
  \sum_{T \in G_3(p)} \tr R(T) = p^{3n} + p \, p^n + (p+1) p^{2n} = p^n (p^n + 1) (p^n + p),
\end{equation}
and we also obtain the following formula, which we record for future reference:
\begin{align}
\label{eq:inner product}
  \frac 1 {p^{3N}} \tr R(T_x) R(T_y)^\dagger &= p^{-d(T_x, T_y)}, \\
\intertext{where}
  d(T_x, T_y) &= \begin{cases}
    0 &\quad \text{if } T_x = T_2 \\
    1 &\quad \text{if } T_x T_y^{-1} \text{ is odd} \\
    2 &\quad \text{if } T_x T_y^{-1} \text{ is even and } T_x \neq T_y
  \end{cases} \nonumber
\end{align}
We note that $d(T_x, T_y)$ defines a metric on $G_3(p)$.
At last we compute the third moment of a random stabilizer state:

\begin{proof}[Proof of formula~\eqref{eq:third} for the third moment]
As explained at the beginning of this section, we can evaluate the third moment of a random stabilizer state by averaging over the Clifford group:
\[ M_3 := \bigl\langle\ket V\!\!\bra V^{\ot 3}\bigr\rangle = \bigl\langle U^{\ot 3} \ket0\!\!\bra0^{\ot 3n} (U^\dagger)^{\ot 3n} \bigr\rangle \]
Here, $\ket V\!\!\bra V$ denotes a stabilizer state and $U$ a Clifford unitary, each chosen uniformly at random.
It is apparent from the right-hand side that $M_3$ commutes $\Cliff(n,p)^{\ot 3}$.
By \cref{thm:commutant}, we can therefore write $M_3 = \sum_{T \in G_3(p)} \gamma_T R(T)$ for some coefficients $\gamma_T \in \CC$.
Now observe that
\[ R(T) M_3 = \bigl\langle U^{\ot 3} R(T) \ket0\!\!\bra0^{\ot 3n} (U^\dagger)^{\ot 3} \bigr\rangle = M_3 \]
for all $T \in G_3(p)$, where the first identify holds since $R(T)$ commutes with $U^{\ot3}$, and the second because $R(T)\ket 0^{\ot3n} = \ket 0^{\ot3n}$.
It follows that all $\gamma_T$ are equal, and hence that $M_3 \propto \sum_{T \in G_3(p)} R(T)$.
We obtain the desired normalization constant in~\eqref{eq:third} by comparing $\tr M_3 = 1$ with~\eqref{eq:sum of traces}.
\end{proof}

We conclude this section with some remarks on higher moments.
For this, denote by $G_k(p)$ the group of orthogonal and doubly stochastic $k\times k$-matrices $T$ and define $R(T)$ accordingly.
Then \cref{lem:commuting,lem:indep} generalizes readily; we have that $[R(T),U^{\ot k}]=0$ for all Clifford unitaries $U\in\Cliff(n,p)$, and the operators $R(T)$ are linearly independent if $n\geq k$.
For example, $G_4(2)$ contains two kinds of matrices: the subgroup of permutation matrices, which we may identify with $S_4$, as well as the `antipermutations`
\[ \bar S_4=\{\begin{pmatrix}1&1&1&1\\1&1&1&1\\1&1&1&1\\1&1&1&1\end{pmatrix} - \pi : \pi \in S_4 \}, \]
so that $G_4(2)=S_4 \cup \bar S_4$. This shows that the commutant of $\Cliff(n,p)^{\ot 4}$ is in general larger than the span of the permutation representation -- even in the case of qubits --, and confirms that in general multiqubit Clifford unitaries do \emph{not} form a $4$-design~\cite{zhu2015multiqubit}.

\section{Detailed derivation of the GHZ bound}\label{app:ghz}

In this section we give a detailed derivation of \cref{thm:main} which bounds the average number of GHZ states that can be extracted from a random stabilizer network.
As in the main text, let $\zeta$ denote the cyclic permutation $1\mapsto2\mapsto3$, so that
\[ \tr (\Psi_{AB}^{T_B})^3 = \tr \Psi^{\ot 3} R_A(\zeta) R_B(\zeta^{-1}). \]
Here, $R_X(T) = r(T)^{\ot X}$ is a representation of an element $T \in G_3(p)$ on the three-fold copy of the Hilbert space corresponding to some subsystem $X$, where we recall that $G_3(p)$ contains the permutations group $S_3$ as a subgroup.
Explicitly, $1 = T_{0,\text{even}}$, $\zeta=T_{-1,\text{even}}$ and $\zeta^{-1}=T_{p,\text{even}}$, as is apparent from~\eqref{eq:concrete elements}.
Using our formula~\eqref{eq:third} for the third moment of a random stabilizer state, we obtain that
\[
  \Braket{\tr (\Psi_{AB}^{T_B})^3}
= \tr \Braket{\Psi^{\ot 3}} R_A(\zeta) R_B(\zeta^{-1})
= \frac 1 {\prod_{x\in V_b} D_x(D_x+1)(D_x+p)} \tr \left(\prod_e \ket e\!\!\bra e^{\ot 3}\right) \left( \prod_{x\in V_b} \sum_T R_x(T)\right) R_A(\zeta) R_B(\zeta^{-1}).
\]
Multiplying out the right-hand product, we find that the above is in turn equal to
\[ \frac 1 {\prod_{x\in V_b} D_x(D_x+1)(D_x+p)} \sum_{\{T_x\}} \tr \left(\prod_e \ket e\!\!\bra e^{\ot 3}\right) \left( \prod_{x\in V} R_x(T_x) \right) \]
where we sum over all assignments $T_x \in G_3(p)$, subject to the boundary conditions that $T_x = \zeta$ for $x\in A$, $T_x = \zeta^{-1}$ for $x\in B$, and $T_x=1$ for $x\in C$.
Now recall that the vertex Hilbert space is a tensor product $\bigotimes_e (\CC^p)^{\ot N}$, where $e$ runs over the edges incident to $x$, and that the representation $R_x(T_x)$ factors correspondingly.
Writing $R_x(T_x) = \bigotimes_e R_{x,e}(T_x)$, we can evaluate the trace edge by edge:
\[ \frac 1 {\prod_{x\in V_b} D_x(D_x+1)(D_x+p)} \sum_{\{T_x\}} \prod_{e=\braket{xy} \in E} \tr \ket e\!\!\bra e^{\ot 3} R_{x,e}(T_x) R_{y,e}(T_y) \]
Any maximally entangled state $\ket{\Phi^+}_{AB}$ satisfies the identity $(X \ot I) \ket{\Phi^+}_{AB} = (I \ot X^t) \ket{\Phi^+}_{AB}$, where $X^t$ denotes the transpose (in the computational basis, i.e., the basis that the maximally entangled state was defined in).
Since $\ket e^{\ot 3}$ is a maximally entangled state on two copies of $(\CC^p)^{\ot 3N}$, we obtain that
\[
  \tr \ket e\!\!\bra e^{\ot 3} R_{x,e}(T_x) R_{y,e}(T_y)
= \frac 1 {p^{3N}} \tr R(T_x) R(T_y)^t
= \frac 1 {p^{3N}} \tr R(T_x) R(T_y)^\dagger
\]
where we write $R(T)$ for the representation of $G_3(p)$ on the three-fold tensor power of $(\CC^p)^{\ot N}$;
the second inequality holds as $R(T)$ is represented by real orthogonal matrices in the computational basis.
According to \cref{eq:inner product}, the right-hand side is given by $p^{-N d(T_x, T_y)}$ and thus we obtain the following fundamental bound:
\begin{equation}
\label{eq:careful}
  \Braket{\tr (\Psi_{AB}^{T_B})^3} \leq p^{-3N_b} \sum_{\{T_x\}} p^{-N \sum_{\braket{xy}} d(T_x,T_y)}
\end{equation}
where the sum is over all choices of $T_x \in G_3(p)$ such that $T_x = \zeta$ in $A$, $T_x = \zeta^{-1}$ in $B$, and $T_x=1$ in $C$.
We note that~\eqref{eq:careful} reduces to~\eqref{eq:careful qubits} in the case of qubits ($p=2)$.

To analyze~\eqref{eq:careful}, we define the \emph{energy} of a configuration by $E[\{T_x\}] := \sum_{\braket{xy}} d(T_x,T_y)$ (cf.\ the main text for a justification of this terminology).
We first consider an arbitrary configuration $\{T_x\}$.
If we denote by $V_A = \{ x : T_x = \zeta \}$ the domain where $T_x$ is assigned the value $\zeta$ then the boundary conditions imply that $V_A\cap V_\partial=A$; that is, $V_A$ is a cut separating $A$ and $\bar A = BC$.
Likewise, the $\zeta^{-1}$-domain $V_B$ is a cut for $B$ and the identity domain $V_C$ a cut for $C$.
These cuts are not necessarily minimal, and so we have that $\lvert\partial V_A\rvert \geq S_{RT}(A)/N$ etc.
Lastly, we write $V' = V_b\setminus (V_A\cup V_B\cup V_C)$ for the remaining vertices and decompose the set of edges into (i) the set of edges $E_1$ that connect any of the domains $V_A$, $V_B$ or $V_C$ with $V'$, (ii) the set of edges $E_2$ that go between any two of the domains $V_A$,$V_B$, and $V_C$, and (iii) the remaining edges $E'$.
We can now lower-bound the energy of the configuration as follows:
\begin{align*}
  E[\{T_x\}]
= \sum_{\braket{xy}\in E_1} d(T_x,T_y) + \sum_{\braket{xy}\in E_2} d(T_x,T_y) + \sum_{\braket{xy}\in E'} d(T_x,T_y)
\geq \lvert E_1 \rvert + 2 \lvert E_2 \rvert
\end{align*}
Indeed, the edges $\braket{xy} \in E_1$ are by definition such that $T_x\neq T_y$, hence $d(T_x,T_y)\geq1$, for the edges in $E_2$ we in addition know that $T_x$ and $T_y$ are even, so that $T_x T_y^{-1}$ is even and hence $d(T_x,T_y)\geq2$.
Furthermore, it is clear that
\[ \lvert E_1 \rvert + 2 \lvert E_2 \rvert = \lvert\partial V_A\rvert + \lvert\partial V_B\rvert + \lvert\partial V_C\rvert \]
since the right-hand side double-counts precisely those edges in $E_2$.
Together, we find that
\[ E[\{T_x\}] \geq E_0 := \bigl( S_{RT}(A)+S_{RT}(B)+S_{RT}(C) \bigr)/N. \]
Moreover, equality holds if and only if the domains $V_A$, $V_B$ and $V_C$ are disjoint minimal cuts for $A$, $B$ and $C$, respectively, and if each connected components of $V'$ is assigned an odd element of $G_3(p)$.
It follows from \cref{lem:disjoint} below that it is always possible to find disjoint minimal cuts for disjoint boundary regions; hence $E_0$ is achievable.
Moreover, if we denote the number of minimal cuts for a boundary region $A$ by $\#_A$ and the maximal number of connected components of any subgraph $V'$ obtained by removing minimal cuts by $\#_b$, then we find that there are at most $\# = (p+1)^{\#_b} \#_A \#_B \#_C$ many configurations of energy $E_0$, for there are $p+1$ odd elements in $G_3(p)$.
All other configurations have higher energy and hence are penalized by a factor of at least $1/p^N$ in~\eqref{eq:careful}.
Thus we obtain the upper bound:
\[
  \Braket{\tr (\Psi_{AB}^{T_B})^3}
\leq p^{-3N_b} p^{-NE_0} \left( \# + \delta \right)
= p^{-3N_b-\left( S_{RT}(A)+S_{RT}(B)+S_{RT}(C) \right)} \left( \# + \delta \right)
\]
where $\delta = (2p+2)^{V_b} / p^N$, since there are no more than $\lvert G_3(p) \rvert^{V_b} = (2p+2)^{V_b}$ non-minimal configurations, and hence
\begin{equation}
\label{eq:log moment upper bound}
\log_p \Braket{\tr (\Psi_{AB}^{T_B})^3}
\leq -3N_b - \bigl( S_{RT}(A)+S_{RT}(B)+S_{RT}(C) \bigr) + \log_p \# + 2 \delta.
\end{equation}
At last we can bound the average number of GHZ states that can be extracted from a random stabilizer network state.
Using~\eqref{eq:ghz no} and $\rho=\Psi/\tr\Psi$, we obtain that
\begin{align*}
  \braket{g}_{\neq 0}
&\leq S_{RT}(A) + S_{RT}(B) + S_{RT}(C) + \log_p \Braket{\tr (\Psi_{AB}^{T_B})^3}_{\neq0} - 3 \Braket{\log_p \tr\Psi}_{\neq0} \\
&\leq S_{RT}(A) + S_{RT}(B) + S_{RT}(C) + \log_p \Braket{\tr (\Psi_{AB}^{T_B})^3} +2\delta + 3 N_b \\
&\leq \log_p \# + 4\delta
\end{align*}
where the first inequality uses $S(X)\leq S_{RT}(X)$ and concavity of the logarithm, the second that $\Pr(\Psi\neq0)\geq1-\delta$ (\eqref{eq:prob trace nz} in \cref{app:rt}), $\tr\Psi\geq1/p^{N_b}$ if $\Psi\neq0$ (\cref{app:trace}) and that $\delta$ is sufficiently small, and the last is obtained by plugging in~\eqref{eq:log moment upper bound}.
This is the statement of \cref{thm:main}.

\begin{lem}
\label{lem:disjoint}
  Let $A$ and $B$ be denote disjoint subsets of $V_\partial$, $V_A$ and $V_B$ minimal cuts for $A$ and $B$, respectively, and $V_0:=V_A\cap V_B$.
  Then either $V_A\setminus V_0$ is a minimal cut for $A$ or $V_B\setminus V_0$ is a minimal cut for $B$.
\end{lem}
\begin{proof}
  Since $V_0\cap V_\partial=\emptyset$, it is clear that $V_A\setminus V_0$ is again a cut for $A$ and $V_B\setminus V_0$ again a cut for $B$.
  We now use that the cut function $c(W) := \lvert\partial W\rvert$ is symmetric and submodular, a fact that is well-known in graph theory.
  It follows that
  \[ \lvert\partial V_A\rvert + \lvert\partial V_B\rvert \geq \lvert\partial(V_A \setminus V_0) \rvert + \lvert\partial(V_B \setminus V_0)\rvert, \]
  and hence that either $\lvert\partial V_A\rvert \geq \lvert\partial(V_A \setminus V_0)\rvert$ or $\lvert\partial V_B\rvert \geq\lvert\partial(V_B \setminus V_0)\rvert$.
  This implies the claim.
\end{proof}

\end{widetext}
\end{document}